\documentclass[10pt]{article}

\usepackage{authblk}
\usepackage{amsfonts,amsmath,amsthm}
\usepackage{tabularx}
\usepackage{graphicx}
\usepackage{adjustbox}
\usepackage{multirow}
\usepackage{enumitem}
\usepackage{graphicx}
\usepackage{url}
\usepackage{color}

\definecolor{maroon}{rgb}{.69,.188,.376}
\definecolor{darkgreen}{rgb}{0,.5,0}
\definecolor{darkblue}{rgb}{0,0,.5}
\definecolor{magenta}{rgb}{1,0,1}

\usepackage[mathscr]{euscript}		

\usepackage{dsfont}

\usepackage{psfrag}			
\usepackage[colorlinks=true]{hyperref}
\hypersetup{pdftex, colorlinks=true, linkcolor=maroon, citecolor=maroon,
  filecolor=blue,urlcolor=blue}

\newtheorem{thm}{Theorem}
\newtheorem{theorem}[thm]{Theorem}
\newtheorem{lemma}[thm]{Lemma}

\theoremstyle{remark}
\newtheorem{remark}{Remark}

\newcolumntype{C}[1]{>{\Centering}p{#1}}
\newcommand{\rupee}{Rs.}

\usepackage{color}
\definecolor{Red}{rgb}{1,0,0}
\definecolor{Blue}{rgb}{0,0,1}
\definecolor{Olive}{rgb}{0.41,0.55,0.13}
\definecolor{Yarok}{rgb}{0,0.5,0}
\definecolor{Green}{rgb}{0,1,0}
\definecolor{MGreen}{rgb}{0,0.8,0}
\definecolor{DGreen}{rgb}{0,0.55,0}
\definecolor{Yellow}{rgb}{1,1,0}
\definecolor{Cyan}{rgb}{0,1,1}
\definecolor{Magenta}{rgb}{1,0,1}
\definecolor{Orange}{rgb}{1,.5,0}
\definecolor{Violet}{rgb}{.5,0,.5}
\definecolor{Purple}{rgb}{.75,0,.25}
\definecolor{Brown}{rgb}{.75,.5,.25}
\definecolor{Grey}{rgb}{.7,.7,.7}
\definecolor{Black}{rgb}{0,0,0}

\title{COVID-19: Optimal Design of Serosurveys for Disease Burden Estimation}

\author[1]{Siva Athreya\thanks{Corresponding Author: athreya@isibang.ac.in}}
\author[2]{Giridhara R. Babu}
\author[3]{Aniruddha Iyer}
\author[3]{Mohammed Minhaas B.S.}
\author[3]{Nihesh Rathod}
\author[3]{Sharad Shriram}
\author[3,4]{Rajesh Sundaresan}
\author[3]{Nidhin Koshy Vaidhiyan}
\author[3]{Sarath Yasodharan}
\affil[1]{Indian Statistical Institute, Bangalore Centre}
\affil[2]{Indian Institute of Public Health, Bangalore}
\affil[3]{Indian Institute of Science, Bangalore}
\affil[4]{Strand Life Sciences}

\begin{document}

\maketitle
\begin{abstract}
We provide a methodology by which an epidemiologist may arrive at an
optimal design for a survey whose goal is to estimate the disease
burden in a population. For serosurveys with a given  budget of $C$
rupees, a specified set of tests with costs, sensitivities, and
specificities, we show the existence of optimal designs in four
different contexts, including the well known c-optimal
design. Usefulness of the results are illustrated via numerical
examples. Our results are applicable to a wide range of
epidemiological surveys under the assumptions that the estimate's
Fisher-information matrix satisfies a uniform positive definite
criterion.
\end{abstract}


\section{Introduction}

For {\em evidence-based} public health management during a pandemic, epidemiologists conduct surveys to estimate the total disease burden. The surveys use multiple tests on participants to gather evidence of both past infection and active infection. For example, the recently concluded first round of the Karnataka state COVID-19 serosurvey \cite{202012medrxiv_BabEtAl} involved three kinds of tests on participants: the serological test on serum of venous blood for detection of immunoglobulin G (IgG) antibodies, the rapid antigen detection test (RAT), and the quantitative reverse transcription polymerase chain reaction (RTPCR) test for detection of viral RNA. The tests provide different kinds of information, e.g. past infection or active infection. The tests have remarkable variation in reliability, e.g., the RAT has sensitivity (i.e probability of a positive test given that the patient has the disease.
) of about 50\% while the RTPCR test has a sensitivity of 95\% or higher (assuming no sample collection- or transportation-related issues); there is variation, albeit less spectacular, in specificities (i.e probability of a negative test given that the patient is well.
). The test costs also vary, e.g., the RAT cost was \rupee\,450 while the RTPCR test cost was \rupee\,1600 on 10 December 2020. Conducting all three tests on each participant is ideal but can lead to high cost.

Suppose $J$ is  the set of tests and  $T = 2^J \setminus \emptyset$ be all the nontrivial subsets of tests. If a subset of tests $t \in T$ is used on a participant, the cost $c_t$ is the sum of the costs of the individual tests constituting $t$. A survey design is defined as $(w_t, t \in T)$, where $w_t \geq 0$ is the number of participants who are administered the subset $t$ of tests. We relax the integrality requirement\footnote{Integrality can be met if we randomise and constrain the expectation.} on $w_t$ and will assume $w_t \in \mathbb{R}_+$. The cost of this design is obviously $\sum_{t \in T} w_t c_t$. The question of interest is:

\begin{quotation}
\noindent {\em Given a budget of C rupees, a set of tests, their costs, their sensitivities, specificities, and the goal of estimating the disease burden, what is a good survey design?}
\end{quotation}

In this short note we address the above question for four specific optimisation criteria. Motivated by the sero-survey discussed above we assume a specific model (see Section \ref{sec:model}) and find the best design for each criteria. In our first result, Theorem \ref{thm:allocation-reduction}, we show the existence of a c-optimal design. We then consider a worst-case design in Theorem \ref{thm:worst-case-design}, optimal design of surveys conducted across different strata in Theorem \ref{thm:weighted}, and optimal design when additional
information about the participants such as symptoms are observed in Theorem \ref{thm:symptom}. We also present numerical examples\footnote{The source code (in R and Python) that was used to generate the results corresponding to Theorems \ref{thm:allocation-reduction} and \ref{thm:worst-case-design} in Table \ref{tab:numericalExamples} can be found at the following URL:
\url{https://github.com/cni-iisc/optimal-sero-survey-design}} illustrating the applicability of the theorems, see in particular Table \ref{tab:numericalExamples} in Section \ref{sec:num}. The results are not limited to the assumed settings and can be applied to more general epidemiological surveys as well, see Remark \ref{rem:gen}.

\subsection{Model and Main Results} \label{sec:model}

\begin{figure}
  \centering
  \includegraphics[scale=0.5]{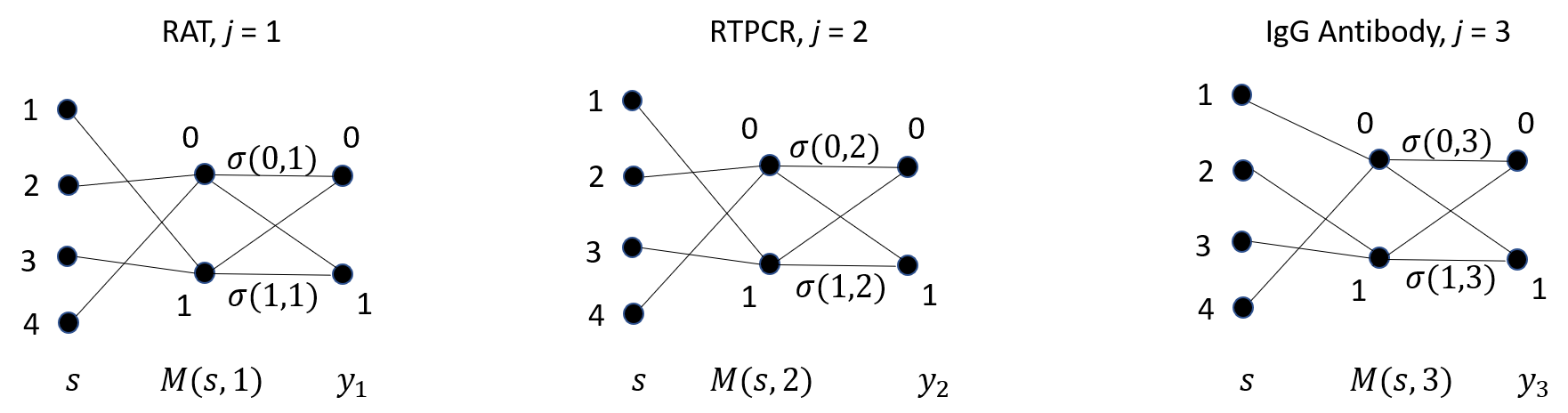}
   \caption{\small Noise model for the RAT, RTPCR, and antibody test outcomes. The left-most subfigure is for RAT, $j=1$. The middle subfigure is for the RTPCR test, $j=2$. The right-most figure is for the antibody test, $j=3$. In each subfigure, the first connection is a ``noiseless'' channel indicated by $M(s,j)$. In each subfigure, the second channel is a noisy binary asymmetric channel whose correct outcome probabilities are given by the specificity $\sigma(0,j)$ and sensitivity $\sigma(1,j)$. The values are indicated in Table \ref{tab:sens-spec}. The cross-over probabilities are 1 minus these.}
   \label{fig:channel}
\end{figure}

As discussed in the introduction, the primary motivation for this note  was to find optimal designs when one needs to use multiple tests to estimate COVID-19 burden in a region. We shall formulate our results with the model used in the Karnataka serosurvey, see \cite{202012medrxiv_BabEtAl}. As we will see in Remark \ref{rem:gen} our results have wider applicability to other models and surveys.

Each individual of the population will
be assumed to be in one of four states. Namely: (i) having active
infection but no antibodies, (ii) having antibodies but no evidence of
active infection, (iii) having both antibodies and active infection,
and (iv) having neither active infection nor antibodies. The state of
the participant in a survey is inferred from the RAT, the RT-PCR, and
the antibody test outcomes, or a subset thereof. As the tests have
known sensitivities and specificities their outcomes are thus noisy
observations of the (unknown) disease state of the individual.

The main focus is on estimating the total disease burden, which is
$\wp := p_1 + p_2 + p_3$, where $p_i$ are the probabilities of the
various states and are unknown to the epidemiologist (See Table
\ref{tab:states-nominal-responses}). We will assume the parameter
space $\mathcal{P}$ to be a convex and compact subset of
$[0,1]^3$. Let the indices $j=1, 2$, and $3$ stand for RAT, RTPCR and
antibody tests, respectively. The last three columns of Table
\ref{tab:states-nominal-responses} provide the nominal responses on
the RAT, RTPCR and the antibody tests for each of the four
states. Write $M(s,j)$ for the nominal test outcome for an individual
in state $s$; $M(s,j) = 1$ indicates a nominal positive outcome, and
$M(s,j) = 0$ indicates a nominal negative outcome.

\begin{table}[ht]
    \centering
    \caption{\small Table of states and nominal test responses $M(s,j)$}
    \begin{tabular}{|c|c|p{2.5in}|c|c|c|}
    \hline
    State $s$ & Probability & State description & RAT & RTPCR & Antibody \\
              &             &      & $j=1$ & $j=2$ & $j=3$ \\
    \hline
    $s = 1$ & $p_1$ & Active infection but no antibodies & 1 & 1 & 0 \\
    \hline
    $s = 2$ & $p_2$ & Antibodies present but no evidence of active infection & 0 & 0 & 1 \\
    \hline
    $s = 3$ & $p_3$ & Simultaneous presence of active infection and antibodies & 1 & 1 & 1 \\
    \hline
    $s = 4$ & $\displaystyle 1 - \sum_{i=1}^3p_i$ & Neither active infection nor antibodies & 0 & 0 & 0 \\
    \hline
    \end{tabular}
    \label{tab:states-nominal-responses}
\end{table}

 As alluded to earlier, the actual test outcomes may differ from the nominal outcomes. In Figure \ref{fig:channel}, there are three {\em channels} for the three tests. Each channel has a deterministic binary-output channel followed by a noisy binary asymmetric channel. The deterministic output is $M(s,j)$ for test $j$ on input $s$. Let $\sigma(0,j)$ denote the specificity of test $j$, and let $\sigma(1,j)$ denote the sensitivity of test $j$. The final output of test $j$ is the output of the noisy binary asymmetric channel with this sensitivity and specificity, see Figure \ref{fig:channel}. We can thus view each test outcome (RAT or RT-PCR or antibody) as the output of two channels in tandem, with the disease state as input.

 \begin{table}[ht]
    \centering
    \caption{\small The sensitivities and specificities}
    \label{tab:sens-spec}
    \begin{tabular}{|c|l|l|l|}
    \hline
    $\sigma(m,j)$  & RAT & RTPCR test & Antibody test\\
                   & $j=1$ & $j=2$ & $j=3$ \\
    \hline
    Specificities ($m=0$) & 0.975 & 0.97 & 0.977 \\
    \hline
    Sensitivities ($m=1$) & 0.5 & 0.95 & 0.921 \\
    \hline
    \end{tabular}
\end{table}

 The specificities and sensitivities used for the estimations\footnote{The reliability indicators for the antibody test were measured on the same lot and batch used for the Karnataka sero-survey \cite{202012medrxiv_BabEtAl} and those for the other tests are the same as those used therein.} are in Table \ref{tab:sens-spec}.  Let $t = (t_1, t_2, t_3)$ represent the subset of tests carried out, or a {\em test pattern}. If $t_j = 1$, then test $j$ had a valid outcome; if $t_j = 0$, then test $j$ was either not conducted or had an invalid outcome. Let $y = (y_1, y_2, y_3)$, where $y_j \in \{ 0, 1, \textsf{NA}\}$. The three values  indicating a negative outcome, a positive outcome or an invalid outcome, respectively.   Let $\mathcal{Y}_t$ denote the set of possible outcomes for the subset of tests $t$.

 We assume that the test outcomes of an individual are conditionally independent given the disease state of the individual\footnote{The assumption of independence is a strong one. Indeed it is possible that a symptomatic person can get tested earlier and during the infectious state, leading to a higher predictive value for the test. For simplicity, we do not model these complications.}. Then for an individual in state $s$ with test pattern $t = (t_1, t_2, t_3)$, the conditional probability of test outcome $y = (y_1, y_2, y_3)$ is given by
\begin{equation}
  \label{eqn:conditionalprobability}
  q(y|s,t) = \prod_{j : t_j = 1} [\sigma(M(s,j),j)]^{{\bf 1}\{M(s,j) = y_j \} } \cdot [1 - \sigma(M(s,j),j))]^{1 - {\bf 1}\{M(s,j) = y_j \}}
\end{equation}
where $M$ and $\sigma$ are given in Tables \ref{tab:states-nominal-responses} and \ref{tab:sens-spec}, respectively.

This leads to a parametric model for the probabilities of test outcomes (observations) given the disease-state probabilities (parameters of the model) denoted $p = (p_1, p_2, p_3)$. Then total burden will be a linear combination $u^T p$ (i.e. $u =(1,1,1)$). The Fisher information matrix, denoted $I_t(p)$, captures information on the reliability of the outcomes when a subset $t$ of tests is conducted on the participant. The $(i,j)$th entry of $I_t(p)$ for the above model can be expressed as
\begin{align}
\sum_{y \in \mathcal{Y}_t} \frac{( q(y|i,t) - q(y|4,t))( q(y|j,t) - q(y|4,t))^T}{\sum_{s=1}^3 p_s q(y|s,t) + (1-p_1-p_2-p_3)q(y|4,t)}. \label{eqn:fisher-information-matrix}
\end{align}

For a given survey design $(w_t, t \in T)$, using the above likelihood function the maximum likelihood estimate (MLE) can be obtained, say $\hat{p}$. Standard results ({e.g. \cite[Chapter 4]{poor-estimation}}) then will assert the consistency and asymptotic normality of the MLE, $\hat{p}$,  as the number of samples goes to infinity while keeping the survey design proportion fixed. We are now ready to state our main results.

\noindent 1. {\em The c-optimal design}.  By the asymptotic optimality property, the estimated total burden $u^T \hat{p}$ is approximately Gaussian with mean $u^T p$ and variance $u^T \left( \sum_{t \in T} w_t I_t(p) \right)^{-1} u$. This suggests the use of the so-called {\em local c-optimal} design criterion. {The adjective `local' refers to local optimality at $p$ and the criterion is called c-optimal because the variance of $c^T \hat{p}$ was minimised, for a vector $c$, in \cite[Chapter~5]{pronzato-pazman-13}.}. Accuracy of the estimate is generally a combination of precision (narrow confidence intervals) and unbiasedness. Broadly speaking, given a budget, c-optimal design attempts to resolve the tension between the accuracy of the estimate and cost efficiency. We shall refer to a survey design $(w^*_t, t \in T)$ as c-optimal if it minimizes the variance $$u^T \left( \sum_{t \in T} w_t I_t(p) \right)^{-1} u$$ given the budget $C$. Our first result is on {\em the c-optimal design}.

\begin{theorem}
  \label{thm:allocation-reduction}
  Let the budget $C>0$, let the parameter $p$ be fixed.
  Then the c-optimal survey
design $(w^*_t, t \in T)$ is given by
\begin{equation}
\label{eqn:w^*_t}
w^*_t = v^*_t C/c_t, \quad t \in T,
\end{equation}
where $v^*$ is a solution to the optimisation problem:
\begin{align}
\text{minimise} & \qquad  a(v;p) := u^T \left( \sum\limits_{t \in T} v_t I_t(p)/c_t \right)^{-1} u  \label{eqn:objective-reduction}\\
\text{subject to} & \qquad \sum\limits_{t \in T} v_t = 1, \quad v_t \geq 0, \, t \in T. \nonumber
\end{align}
Furthermore, the minimum variance is given by $a(v^*;p)/C$.
\end{theorem}

In other words, for c-optimality, we must allocate $v^*_t$ fraction of
the budget $C$ to the subset of tests $t$; the number of participants
$w^*_t$ that should be administered the subset $t$ of tests is
therefore given by \eqref{eqn:w^*_t}. Note that $v^*$ is independent
of the budget $C$. A subset $t$ may be highly informative, i.e., $u^T
I_t(p)^{-1} u$ may be small, yet $t$ may not be attractive when
relativised by its cost $c_t$ because what ultimately matters is $u^T
(I_t(p)/c_t)^{-1} u$; see \eqref{eqn:objective-reduction}. Section
\ref{sec:proof-A} contains the proof with additional structural
results on the support of the solution. The last statement of the
theorem gives the level of accuracy\footnote{We will assume that the bias is negligible.} (asymptotic variance of the
re-scaled estimate) that the budget can buy: accuracy is inversely
proportional to the budget. It also suggests how to solve a related
problem: what is the required budget to meet a certain margin of error
target with confidence $1-\alpha$? The answer to this question is to
find the $C$ that lowers $\Phi^{-1}(\alpha/2) \sqrt{a(v^*; p)/C}$
below the target margin of error, where $\Phi$ is the normal
cumulative distribution function.

\noindent 2. {\em Worst-case design.} Theorem \ref{thm:allocation-reduction} is a locally optimal design at parameter $p$, and the optimal design $v^*$ may depend on $p$ in general. Whenever we want to highlight this dependence, we will write $v^*(p)$. In practice, we often design for a guessed $p$; for e.g., the Karnataka serosurvey \cite{202012medrxiv_BabEtAl} assumed a seroprevalence of 10\% prevalence to arrive at the required number of antibody tests. One could also use a prior for $p$ and extend Theorem \ref{thm:allocation-reduction} straightforwardly. Our next result, however, is a {\em design for the worst case}, i.e., a design $(w_t, t \in T)$ that minimises $$\max\limits_{p \in \mathcal{P}} u^T \left( \sum\limits_{t \in T} w_t I_t(p) \right)^{-1} u,$$ where $\mathcal{P}$ is a convex and compact uncertainty region for the parameter, subject to the budget constraint $C$.

As optimal designs for the worst case usually go, our solution will involve the Nash equilibria of a two-player zero-sum game $G$ defined as follows: The pay-off of the game $G$ is $$u^T (\sum\limits_{t \in T} v_t I_t(p)/c_t)^{-1} u$$ with the minimising player chooses $v$ from the probability simplex on $T$; the maximising player simultaneously chooses $p$ from a convex and compact $\mathcal{P}$. The design for the worst case is as follows.

\begin{theorem}
\label{thm:worst-case-design}
  Any design $w^*$ for the worst case is of the form $w^*_t = v^*_t C/c_t$, $t \in T$, where $v^*$ is a Nash equilibrium strategy for the minimising player in the game $G$.
\end{theorem}

As we shall see, the strategy sets are convex and compact, the pay-off function is concave in $p$ for a fixed $v$, convex in $v$ for a fixed $p$, and by Glicksberg's theorem \cite{glicksberg-52} there exists a Nash equilibrium. While the equilibrium may not be unique, it is well-known that this is not an issue for two-player zero-sum games. The minimising player may pick any Nash equilibrium and play his part of the equilibrium, i.e., there is no ``communication problem'' in a two-player zero-sum game.

\noindent 3. {\em Allocation across strata}. Surveys are also conducted across different strata or districts with differing disease burdens. For example, consider $D$ districts, with $p(d)$ being the disease spread parameter in district $d$, $d \in [D]= \{1, \ldots, D\}$. If the population fraction of district $d$ is $n_d$, with $\sum_{d \in [D]} n_{d} = 1$, then the weighted overall disease burden in the $D$ districts is $\sum_{d \in [D]} n_d (u^T \hat{p}(d))$ where $\hat{p}(d)$ is the MLE for district $d$ when the true parameter is $p(d)$. The next result indicates how the total budget must be allocated across the $D$ districts to minimise the variance of $\sum_{d \in [D]} n_d (u^T \hat{p}(d))$, subject to a total budget of $C$.

\begin{theorem}
\label{thm:weighted}
Let the parameters for the districts be $p(d)$ and let the population fractions be $n_d$, $d = 1, \ldots, D$.
The (local) c-optimal design for the weighted estimate $\sum_{d \in [D]} n_d (u^T \hat{p}(d))$ is as follows. The optimal budget allocation $C_d$ for district $d$ is proportional to $$n_d \sqrt{a(v^*(p(d));p(d))},$$ i.e.,
\[
  C_d = C \cdot \frac{n_d \sqrt{a(v^*(p(d));p_d)}}{\sum_{d' \in [D]} n_{d'} \sqrt{a(v^*(p(d'));p(d'))}}, \quad d = 1, \ldots, D,
\]
where $v^*(p(d))$ solves \eqref{eqn:objective-reduction} for district $d$. The optimal design in each district is given by $w_t^*(d) = v_t^*(p(d))C_d /c_t$ for each $t$ and $d$.
\end{theorem}

Thus a larger fraction of the total budget should be allocated to districts with a larger population, and to districts with a greater uncertainty in the estimate as measured by the standard error. The precise mathematical relation is the content of the above theorem. If there is no basis for a guess of $p(d)$ in each district, then the worst-case design could be employed in each district. Theorem \ref{thm:weighted} then implies that the budget allocation is in the proportion of the population. Theorem \ref{thm:weighted} is also applicable to other kinds of stratification, for e.g., sex, age groups, risk categorisation, pre-existing health conditions, etc., so long as the fraction in each component of the stratification under consideration\footnote{For example, we should know the male fraction of the population and the female fraction the population for the sex-based stratification. These replace the $n_d$ in Theorem \ref{thm:weighted}.} is known.

\noindent 4. {\em Enabling on-the-ground decisions}. Often the surveyors have access to additional information about the participants that might affect the quality of the test outcomes. Let us consider the example of participants presenting with or without symptoms. This is useful information that can inform what tests should be employed. The Karnataka serosurvey found that the RAT is more sensitive on participants with symptoms than on participants without symptoms ($68\%$ and $47\%$, respectively, see \cite{202012medrxiv_BabEtAl}). Let $s$ by the symptom-related observable with $s = 0$ denoting the absence of symptoms and $s=1$ denoting the presence of symptoms. Let $r_s$, $s=0,1$ denote the asymptomatic and the symptomatic fractions of the population. Let $p(s)$ be the model parameter for substratum $s$, $s=0,1$. The goal is to estimate the total burden $\sum_{s=0,1} r_s (u^T p(s))$. Assuming MLE in each substratum, the objective continues to be the minimisation of the variance of $\sum_{s=0,1} r_s (u^T \hat{p}(s))$.

Since the tests' reliability depend on whether a participant is symptomatic or asymptomatic, let $I_t^s(p(s))$ be the Fisher information matrix associated with the group $s$, when the subset of tests employed is $t$ and the group $s$ parameter is $p(s)$. Let $v^*(p(s))$ denote the optimum allocation of the budget\footnote{Note that a more proper notation would have been $v^*(p(s);s)$ at the expense of cumbersome notation. We will omit the dependence of $v^*$ on $s$ for notational convenience.}, given to group $s$, across the subsets of tests. Further, let $a_s(v^*(p(s));p(s))$ denote the value of \eqref{eqn:objective-reduction} for group $s$, when the parameter is $p(s)$ and the Fisher information matrices are $I_t^s(p(s))$, $t \in T$. The next result is the following.

\begin{theorem}
  \label{thm:symptom}
  The optimal budget allocation $C_s$ for group $s$ is in proportion to $$r_s \sqrt{a_s(v^*(p(s));p(s))},$$ i.e.,
\[
  C_s = C \frac{r_s \sqrt{a_s(v^*(p(s));p(s))}}{\sum\limits_{s' =0,1} r_{s'} \sqrt{a_{s'}(v^*(p(s'));p(s'))}}, \quad s = 0,1,
\]
where $v^*(p(s))$ solves \eqref{eqn:objective-reduction} for group $s$. The optimal design for each group $s$ is given by $w_t^*(s) = v_t^*(p(s))C_s/c_t$ for each $t$ and $s$.
\label{thm:allocation-symptomatic-model}
\end{theorem}

Note that the optimal allocation in Theorem \ref{thm:allocation-symptomatic-model} is similar to that in Theorem \ref{thm:weighted}, with the difference being $a_s(v^*(p(s));p(s))$ depending on the group $s$ in Theorem \ref{thm:allocation-symptomatic-model}. The optimal budget for $C_0$ and $C_1$ suggests how many asymptomatics and symptomatics must be tested. Furthermore, the test prescription adapts to the presence or absence of symptoms presented by the participant. {Finally, Theorem~\ref{thm:symptom} is also applicable to other settings where information about participants' variables and factors are observable, and are known to affect the quality of testing, e.g., severity of the symptoms and the possible high viral load leading to better sensitivity of viral RNA detection tests.}

\subsection{Numerical Examples} \label{sec:num}
In this section we shall discuss several numerical examples illustrating the applications of our main results.

We begin with some examples connected to Theorem
\ref{thm:allocation-reduction}. Consider a local c-optimal design at
$p = (0.10, 0.30, 0.01)$, the fractions of those with active infection
alone, those with antibodies alone, and those that show both active
infection and antibodies. Suppose that the budget is
\rupee\,1,00,00,000 = \rupee\,1 Crore. Let us fix the RAT and antibody
test costs at \rupee\,450 and \rupee\,300, respectively. In the first
row of Table \ref{tab:numericalExamples}, the RTPCR test costs
\rupee\,1,600. The optimal allocation is to run antibody tests alone
on 521 participants ($t = (0,0,1)$) and the RAT and antibody tests on
13,125 participants ($t = (1,0,1)$). In particular, no RTPCR tests are
done. In the second row of Table \ref{tab:numericalExamples}, the
RTPCR test cost comes down to \rupee\,1000. Then the optimal
allocation is to run antibody tests ($t = (0,0,1)$) on a larger number
of participants, namely 8,000, and the RTPCR and antibody tests on
5,846 participants ($t = (0,1,1)$). In particular, RTPCR is
sufficiently competitive that no RAT tests are done. In the third row
of Table \ref{tab:numericalExamples}, the RTPCR test costs only\footnote{This big reduction in cost may come from pooled testing strategies.}
\rupee\,100. The optimal allocation is to run exactly one test
pattern, the RTPCR and the antibody tests ($t = (0,1,1)$), on 25,000
participants.

\begin{table}[t]
    \centering
    \caption{\small Numerical examples for a budget of \rupee\,1 Crore.}
    {\tiny
\begin{tabular}{|l|r|r|p{0.25in}|c|c|c||c|c|c|c|}
  \hline
  {} & \multicolumn{3}{c|}{Test cost in \rupee} & \multicolumn{3}{c||}{Parameter or parameter range} & \multicolumn{4}{c|}{Optimal design (quantised to $\mathbb{Z}_+$)} \\
 Criterion & RAT & RTPCR & Anti-body & $p_1$ & $p_2$ & $p_3$ & S/A & (0,0,1) & (1,0,1) & (0,1,1) \\
  \hline
  \hline
  Local & 450 & 1,600 & 300 & 0.10 & 0.30 & 0.01 & {-} & 521 & 13,125 & 0 \\
  \hline
  Local & 450 & 1,000 & 300 & 0.10 & 0.30 & 0.01 & {-} & 8,000 & 0 & 5,846 \\
  \hline
  Local & 450 & 100   & 300 & 0.10 & 0.30 & 0.01 & {-} & 0 & 0 & 25,000 \\
  \hline
  Worst-case & 450 & 100   & 300 & 0.01-0.15 & 0.10-0.50 & 0.00-0.02 & {-} & 838 & 0 & 24,371 \\
  \hline
  On-the-ground, 10\% symptomatic & \multirow{2}{*}{450} & \multirow{2}{*}{1,600}   & \multirow{2}{*}{300} & \multirow{2}{*}{0.10} & \multirow{2}{*}{0.30} & \multirow{2}{*}{0.01} & S & 300 & 1,046 & 0 \\
   RAT sensitivity 0.68(S) 0.47(A) &     &       &     &           &
            &           & A & 0 & 12,167 & 0 \\
  \hline
\end{tabular}
}
\label{tab:numericalExamples}
\end{table}

Our next set of example is with regard to Theorem
\ref{thm:worst-case-design}. Consider the numerical example in the
fourth row of Table \ref{tab:numericalExamples} where the RTPCR test
cost is only \rupee\,100. The range for $p_2$ comes from the
Karnataka's serosurvey \cite{202012medrxiv_BabEtAl} where the
districts' total burden varied from 8.7\% to 45.6\%. The ranges for
$p_1$ and $p_3$ also come from \cite{202012medrxiv_BabEtAl}. The
worst-case (obtained via a grid search with 0.01 increments) is near
$p = (0.06, 0.45, 0.0)$. The optimal design is to spend 2.5\% of the
budget on antibody tests alone, and 97.5\% of the budget on the
combination of RTPCR and antibody tests. For the indicated costs, this
comes to 838 participants getting only antibody tests and 24,371
participants getting the RTPCR and antibody tests, see the fourth row
of Table \ref{tab:numericalExamples}.

Since an overwhelming part of the budget is for the test pattern $(0,1,1)$, for logistical reasons,  all participants will be administered RTPCR and antibody tests in the second serosurvey for Karnataka (January 2021). Calculations (that we do not present here) indicate that the worst-case variance increases by a factor 1.0023, and so the number of samples must be increased by this factor to achieve the same margin of error and confidence. But the logistical simplicity is well worth extra cost, considering the magnitude and decentralised nature of the survey.

{We now discuss an insightful numerical example related to Theorem \ref{thm:allocation-symptomatic-model}. Consider a setting where the RAT costs \rupee\,450, the RTPCR costs \rupee\,1,600 and the antibody test costs \rupee\,300. Let us go back to the local c-optimal design criterion for $p = (0.1, 0.3, 0.01)$. Additionally, let 10\% of the population be symptomatic ($r_1 = 0.1$). The optimal allocations under Theorem \ref{thm:allocation-symptomatic-model} indicate that roughly 91.2\% of the budget must be spent on the asymptomatic population and 8.8\% on the symptomatic population. The reduction for the symptomatic population is because the estimate is more reliable (given the better sensitivity of the RAT). As expected, the RTPCR test is too expensive and is completely skipped. More interestingly, all asymptomatic participants should be administered the RAT and the antibody tests whereas among the symptomatic participants, a sizable fraction should be administered the antibody test alone. In other words, a smaller fraction of the symptomatic fraction get the RAT. This might be paradoxical at first glance, for one might expect an increased use of RAT on participants where it is more effective. There is however an insightful resolution to this paradox: given that the active infection fraction estimate is likely to be better in the symptomatic population than in the asymptomatic population, we must invest a portion of the budget for the symptomatic population on antibody tests alone to improve the estimation of the fraction with past infection. In the last row in Table \ref{tab:numericalExamples}, we have indicated the numbers to be tested according to the optimal design.}

\subsection{Prior work and context for our findings}

The locally \emph{c-optimal} design criterion is well-known \cite[Chapter~5]{pronzato-pazman-13}. Elfving's seminal work \cite{elfving-52} shows that, for a linear regression model, the c-optimal design is a convex combination of the extreme points in the design space. Dette and Holland-Letz~\cite{dette-09} extended the characterisation of the c-optimal design to nonlinear regression models with unequal noise variances. Our Theorem \ref{thm:allocation-reduction} could be recovered from the results of Dette and Holland-Letz~\cite{dette-09} with some additional work to characterise the solution.  Chernoff~\cite{chernoff-53} considered the closely related criterion of minimising {the trace of $(\sum_{t \in T} w_t I_t(p))^{-1}$} and obtained a convex-combination-of-extreme-points characterisation. This criterion is useful in our setting if the surveyor is interested in estimating the individual $p_1, p_2$, and $p_3$ with good accuracies, not just their sum, and gives equal importance to the three variances. Other design criteria include an optimality criterion that takes into account the probability of observing a desired outcome in generalised linear models~\cite{mcgree-eccleston-08}, variance minimisation for linear models but with non-Gaussian noise~\cite{gao-zhou-14}, etc. For us, c-optimality is the most appropriate since our interest is in estimating $u^T p$.

While the above works deal with locally optimal designs, we {also} consider the worst-case design problem, weighted design, and design for on-the-ground decisions in Theorems \ref{thm:worst-case-design}, \ref{thm:weighted}, and \ref{thm:symptom}, respectively. The worst-case design problem was considered in~\cite[Chapter~8]{pronzato-pazman-13} where a necessary and sufficient condition for optimality is given. In principle it should be possible to derive Theorem \ref{thm:worst-case-design} from that starting point. Our game-theoretic characterisation was more appealing. To the best of our knowledge, Theorem \ref{thm:symptom} and applications of Theorems \ref{thm:allocation-reduction}-\ref{thm:symptom} to COVID-19 survey design are new. We hope that other survey designers will make use of our findings at this critical juncture of the pandemic, so that immunisation plans can be better optimised.

\begin{remark} \label{rem:gen}

      Our focus in this note has been optimal design of
      serosurveys for disease burden estimation. 
However our results hold for more general epidemiological
surveys. The dimensionality of the parameter $p$ can be any natural number $k$, and the $u$ defining the estimate $u^T p$ can be any vector in ${\mathbb R}^k$.  In any such a model our results will hold as long as the MLE exists and  the corresponding Fisher information matrix, $I_t(p)$, corresponding to the parameter $p$, satisfies the assumption (A1) stated below for Theorem \ref{thm:allocation-reduction}, Theorem \ref{thm:weighted}, and Theorem \ref{thm:symptom}, and the assumption (A2), also stated below, for Theorem \ref{thm:worst-case-design}, see Remark \ref{rem:proof} in the next section.

For $t \in T$ and $p \in \mathcal{P}$, let $\lambda_p(t)$
denote the smallest eigenvalue of $I_t(p)$. We now make precise the assumptions.

\begin{enumerate}[label=({A\arabic*})]
\item For a given $p \in \mathcal{P}$, there exists $t^* \in T$  such that the matrix $I_{t^*}(p)$ is positive definite.
\label{assm:pd}
\item There exists $t^* \in T$  such that $\inf_{p \in \mathcal{P}} \lambda_p(t^*) > 0$.
\label{assm:pd-worst-case}
\end{enumerate}
\end{remark}

\section{Proof of Theorems} \label{sec:proof-A}

Recall that $C>0$ is the total budget, $T $ is all the nontrivial subsets of tests, $c_t$ is the sum of the costs of the individual tests constituting $t \in T$, and the survey design is denoted by $(w_t, t \in T)$, where $w_t \geq 0$ is the number of participants (relaxed to a real-value) who are administered the subset $t$ of tests, $w_t \in \mathbb{R}_+$. The cost of this design is $\sum_{t \in T} w_t c_t$.

We begin with a specific observation with regard to our model.
    \begin{remark} \label{rem:proof}Let $t^* = (1,1,1)$.
      When the specificities and sensitivities are as given by Table~\ref{tab:sens-spec}, it is easy to verify (A1) and (A2) hold. Indeed, from \eqref{eqn:fisher-information-matrix}, defining
      $$
      v(y,t^*) := \left( q(y|1,t^*) - q(y|4,t^*), q(y|2,t^*) - q(y|4,t^*), q(y|3,t^*) - q(y|4,t^*) \right)^T,
      $$
      and the probability
      $$
      P(y|t^*) := \sum_{s=1}^3 p_s q(y|s,t^*) + (1-p_1-p_2-p_3)q(y|4,t^*),
      $$
      we have for any nonzero $x \in \mathbb{R}^3$
      \begin{eqnarray*}
        x^T I_{t^*}(p) x
        & = & x^T \left( \sum_{y \in \mathcal{Y}_{t^*}} \frac{1}{P(y|t^*)} v(y,t^*) v(y,t^*)^T \right) x \\
        & = & \sum_{y \in \mathcal{Y}_{t^*}} \frac{1}{P(y|t^*)} \left|x^T v(y,t^*)\right|^2\\
        & \geq &  \sum_{y \in \mathcal{Y}_{t^*}} \left|x^T v(y,t^*)\right|^2 \\
        & > & 0,
      \end{eqnarray*}
      where the last strict inequality holds because $\{ v(y,t^*), y \in \mathcal{Y}_{t^*} \}$ is linearly independent, a fact that can verified for the given sensitivities and the specificities. This not only verifies assumption (A1) for our example, but also assumption (A2).
\end{remark}

We shall use the fact the model satisfies assumption (A1) or (A2) in the proofs below.

\begin{proof}[Proof of Theorem~\ref{thm:allocation-reduction}]
Let $v_t = w_t c_t /C, t \in T$. Recall $a(v;p)$ from \eqref{eqn:objective-reduction}. It is immediate that a minimiser of \eqref{eqn:objective-reduction} will yield  a c-optimal design given by \eqref{eqn:w^*_t}. To find a minimiser, first  consider the following optimisation problem
\begin{align}
\text{minimise} & \qquad  a(v;p) \label{eqn:objective-reduction-1} \\
\text{subject to} & \qquad \sum\limits_{t \in T} v_t \leq 1, \quad v_t \geq 0, \, t \in T. \label{eqn:objective-reduction-constraints}
\end{align}
{By Assumption~\ref{assm:pd}, there exists a $v^{(0)}$ such that $a(v^{(0)}; p) < \infty$, and hence the minimisation can be restricted to the set $\mathcal{C} = \{(v_t, t \in T): a(v;p) \leq a(v^{(0)}; p), v_t \geq 0, t \in T, \sum_{t \in T} v_t \leq 1\}$.} With this restriction on the constraint set, the above problem is a convex optimisation problem. Indeed, the objective function is convex on $\mathcal{C}$; for any $\lambda \in (0,1)$ and any two designs $v^{(a)}$ and $v^{(b)}$ { in $\mathcal{C}$,} the matrix
\begin{align*}
\lambda \left( \sum\limits_{t \in T} \frac{v^{(a)}_t I_t(p)}{c_t} \right)^{-1} + (1-\lambda) \left( \sum\limits_{t \in T} \frac{v^{(b)}_t I_t(p)}{c_t} \right)^{-1} - \left( \sum\limits_{t \in T} \frac{(\lambda v^{(a)}_t + (1-\lambda) v^{(b)}_t) I_t(p)}{c_t} \right)^{-1}
\end{align*}
is positive semidefinite, and the constraint set {$\mathcal{C}$} is a convex subset of $\mathbb{R}^{|T|}$. Moreover $\mathcal{C}$ is also compact: for any sequence $\{v^{(n)}, n \geq 1\}$ in $\mathcal{C}$, being a sequence in the compact set~(\ref{eqn:objective-reduction-constraints}), we can extract a subsequence that converges to an element $\tilde{v}$ of~(\ref{eqn:objective-reduction-constraints}); from the definition of $\mathcal{C}$, we must have $a(\tilde{v};p) \leq a(v^{(0)};p)$ which implies that $\tilde{v} \in \mathcal{C}$. Hence there exists a $v^* = (v^*_t, t\in T)$ that solves the above problem, see \cite[Section 4.2.2]{bvop}. If $\sum_{t \in T} v_t^* < 1$, then we may scale up the $v^*$ by a factor to use the full budget, satisfy the sum-constraint with equality, and strictly reduce the objective function by the same factor, which is a contradiction to $v^*$'s optimality. Hence $\sum_{t \in T} v_t^* = 1$. Hence $w^*_t = v^*_t C / c_t$ is the desired c-optimal design and the minimum variance is given by $a(v^*;p)/C$. This completes the proof.
\end{proof}

For information on the structure of an optimal solution, consider the Lagrangian
\begin{align*}
L(v, \lambda, \mu) = u^T \left( \sum\limits_{t \in T} v_t I_t(p)/c_t \right)^{-1} u - \sum\limits_{t \in T} \lambda_t v_t + \mu\left(\sum\limits_{t \in T} v_t -1\right).
\end{align*}
For each $t \in T$,
\begin{align*}
\partial_{v_t} L(w, \lambda, \mu) & = - u^T \left( \sum_{t' \in T} v_{t'} I_{t'}(p)/c_{t'} \right)^{-1} \frac{I_t(p)}{c_t} \left( \sum_{t' \in T} v_{t'} I_{t'}(p)/c_{t'} \right)^{-1} u  - \lambda_t + \mu.
\end{align*}
By the Karush–Kuhn–Tucker conditions \cite[Chapter~5]{bvop}, there exist non-negative numbers $(\lambda^*_t, t \in T)$ and $\mu^*$ such that
\begin{align}
  -u^T \left( \sum_{t' \in T} v^*_{t'} I_{t'}(p)/c_{t'} \right)^{-1} \frac{I_t(p)}{c_t}  \left( \sum_{t' \in T} v^*_{t'} I_{t'}(p)/c_{t'} \right)^{-1} & u  - \lambda^*_t + \mu^* = 0, t \in T, \label{eqn:kkt-trace-condition}\\
\lambda^*_t v^*_t= 0, t \in T, \label{eqn:slackness-1} \\
 \mu^*\left( \sum\limits_{t \in T} v^*_t - 1\right)  = 0. \label{eqn:slackness-2}
\end{align}
Clearly $v^* \not \equiv 0$ and $\mu^* > 0$, otherwise \eqref{eqn:kkt-trace-condition} is violated. Conditions~(\ref{eqn:kkt-trace-condition}) and~(\ref{eqn:slackness-1}) together with the fact that $w^*_t = v^*_t C/c_t$ imply that whenever $w^*_t > 0$ for some $t \in T$, we must have
\begin{align*}
 u^T \left( \sum\limits_{t \in T} v^*_t I_t(p)/c_t \right)^{-1} \frac{I_t(p)}{c_t} \left( \sum\limits_{t \in T} v^*_t I_t(p)/c_t \right)^{-1} u  = \mu^*.
\end{align*}
We can compute $\mu^*$ easily: multiplying by $v_t^*$, summing over $t$, and using $\sum_{t \in T} v_t^* = 1$, we see that $$\mu^* = u^T \left( \sum_{t \in T} v^*_t I_t(p)/c_t \right)^{-1} u = a(v^*;p).$$

The proof of Theorem~\ref{thm:worst-case-design} requires a preliminary lemma. Fix $K \geq 3$ and let $v_1, \ldots, v_K \in \mathbb{R}^3$. For $a \in \mathbb{R}^K$ such that $\sum_{i=1}^K a_i = 1$, define
\begin{align*}
I(a) = 	\sum_{i=1}^K \frac{1}{a_i} v_i v_i^T.
\end{align*}
We first prove the following lemma. A similar concavity property holds in the context of parallel sum of positive definite matrices; see~\cite[Theorem~4.1.1]{bhatia}. Let $\mathbb{R}^K_{++}$ denote the set of $K$-vectors whose entries are strictly positive.

\begin{lemma} Assume $I(a)$ is invertible for all $a \in \mathbb{R}^K_{++}$. The mapping $a \mapsto u^T
I(a)^{-1}u$ is concave on $\mathbb{R}^K_{++}$.
\label{lemma:concavity}
\end{lemma}
\begin{proof}
Let $f(a) = u^T I(a)^{-1}u$. We have, for $i = 1,2,\ldots, K$,
\begin{align*}
\partial_{a_i} I(a)^{-1} &= -I(a)^{-1} (\partial_{a_i}I(a)) I(a)^{-1} \\
& = \frac{1}{a_i^2}I(a)^{-1} v_i v_i^T I(a)^{-1}.
\end{align*}
We also have, for $i, j = 1,2,\ldots, K$,
\begin{align*}
\partial^2_{a_i a_j} I(a)^{-1}& = \frac{I(a)^{-1} v_i v_i^T I(a)^{-1} v_j v_j^T I(a)^{-1} + I(a)^{-1} v_j v_j^T I(a)^{-1} v_i v_i^T I(a)^{-1}}{a_i^2 a_j^2} \\
& \qquad + \delta_{i=j} \left(\frac{-2}{a_i^3}\right) I(a)^{-1} v_i v_i^T I(a)^{-1}.
\end{align*}
Hence
\begin{align*}
\frac{1}{2}\partial^2_{a_i a_j} f(a)& = \frac{u^TI(a)^{-1}v_i \times v_i^T I(a)^{-1} v_j \times u^T I(a)^{-1} v_j}{a_i^2 a_j^2} \\
& \qquad - \delta_{i=j}\frac{1}{a_i^3} (u^T I(a)^{-1} v_i)^2.
\end{align*}
Let $\beta_i = u^T I(a)^{-1} v_i$ and let $B_{i,j} = v_i^T I(a)^{-1} v_j$. Then
\begin{align*}
\frac{1}{2}\partial^2_{a_i a_j}f(a) = \frac{\beta_i B_{i,j} \beta_j}{a_i^2 a_j^2} - \delta_{i=j} \frac{1}{a_i^3} \beta_i^2.
\end{align*}
For positive semidefiniteness of the Hessian of $f$, it suffices to show that
\begin{align*}
\sum_{i,j =1}^K \frac{x_i \beta_i B_{i,j}\beta_j x_j}{a_i^2 a_j^2}  - \sum_{i=1}^K \frac{x_i^2 \beta_i^2}{a_i^3} \leq 0
\end{align*}
for all $x \in \mathbb{R}^K$. With $\alpha_i = \frac{x_i \beta_i}{a_i \sqrt{a_i}}$, the above condition becomes
\begin{align*}
\sum_{i,j=1}^K \alpha_i \frac{B_{i,j}}{\sqrt{a_i a_j}} \alpha_j - \sum_{i=1}^K \alpha_i^2 \leq 0,
\end{align*}
which is the same as asking for the largest eigenvalue of the matrix with entries
\begin{align*}
\left( \frac{B_{i,j}}{\sqrt{a_i a_j}} \right)_{i,j = 1,2,\ldots, K}
\end{align*}
to be at most $1$. Note that
\begin{align*}
\frac{B_{i,j}}{\sqrt{a_i a_j}} & = \frac{1}{\sqrt{a_i}} v_i^T I(a)^{-1} v_j \frac{1}{\sqrt{a_j}} \\
& = \left(\frac{v_i}{\sqrt{a_i}}\right)^T I(a)^{-1}\left(\frac{v_j}{\sqrt{a_j}}\right) \\
& = \left(\frac{v_i}{\sqrt{a_i}}\right)^T \left(\sum_{i'} \left(\frac{v_{i'}}{\sqrt{a_{i'}}}\right) \left(\frac{v_{i'}}{\sqrt{a_{i'}}}\right)^T \right)^{-1} \left(\frac{v_j}{\sqrt{a_j}}\right).
\end{align*}
Let $\tilde{v}_i = \frac{v_i}{\sqrt{a_i}}$ and $\tilde{V} = [\tilde{v}_1, \ldots, \tilde{v}_K]$. Then the above entry is the $(i,j)$ entry of the matrix $\tilde{V}^T (\tilde{V} \tilde{V}^T)^{-1} \tilde{V}$. It therefore suffices to show that, in positive semidefinite ordering,
\begin{align*}
\tilde{V}^T (\tilde{V} \tilde{V}^T)^{-1} \tilde{V} \leq I_{K},
\end{align*}
that is, for each $z \in \mathbb{R}^K$, we must have
\begin{align*}
(\tilde{V}z)^T (\tilde{V} \tilde{V}^T)^{-1} \tilde{V}z \leq \|z\|^2.
\end{align*}
Let $\tilde{V}$ have the singular value decomposition $U \Sigma V^T$, where $U$ and $V$ are unitary matrices of appropriate sizes. Since $I(a)$ is invertible, it follows that $(\tilde{V}\tilde{V}^T)^{-1} = U \Lambda^{-1} U^T$ where $\Lambda = \Sigma \Sigma^T$. Also $\Sigma^T (\Sigma \Sigma^T)^{-1} \Sigma$ is a block diagonal matrix with $I_3$ and $O_{K-3}$ (all zero matrix square matrix of dimension $K-3$) on the diagonal. Therefore, the left-hand side of the above display becomes
\begin{align*}
z^T V \Sigma^T U^T U \Lambda^{-1} U^T U \Sigma V^T z & = z^T V \Sigma^T (\Sigma \Sigma^T)^{-1} \Sigma V^T z\\
& = z^T V \begin{bmatrix}
I_3 & O_{3,K-3}\\
O_{K-3,3} & O_{K-3}
\end{bmatrix} V^T z
\end{align*}
where $O_{m,n}$ is the all zero matrix of dimensions $m \times n$. Since $V$ is orthonormal, the lemma follows.
\end{proof}

\begin{proof}[Proof of Theorem~\ref{thm:worst-case-design}]
Consider the zero-sum game $G$. By assumption~\ref{assm:pd-worst-case}, there exists a design $v^{(0)}$ such that $\sup_{p \in \mathcal{P}} a(v^{(0)};p) < \infty$; hence, we may restrict the set of designs to $$\mathcal{C}:=\left\{(v_t, t \in T): v_t \geq 0, t \in T, \sum_{t \in T} v_t = 1, a(v;p) \leq \sup_{p' \in \mathcal{P}}a(v^{(0)};p') \text{ for all }p \right\}.
$$

\noindent Recall that, for a given pair of strategies $(v, p)$, the pay-off of the maximising player is $a(v; p) = u^T \left(\sum_{t \in T} v_t I_t(p)/c_t \right)u$. For each $p \in \mathcal{P}$, as argued in the proof of Theorem~\ref{thm:allocation-reduction}, the mapping $v \mapsto a(v;p)$ is convex on $\mathcal{C}$, and $\mathcal{C}$ is a convex and compact subset of the design space. Recall $q$ from \eqref{eqn:conditionalprobability} and the Fisher information matrix from \eqref{eqn:fisher-information-matrix}. The mapping
\[
  p \mapsto \sum_{s=1}^3 p_s q(y|s,t) + (1-p_1-p_2-p_3)q(y|4,t), \forall y \in \mathcal{Y}_t, \forall t \in T
\]
is linear. So by Lemma~\ref{lemma:concavity}, for each fixed $v \in \mathcal{C}$, the mapping $p \mapsto a(v;p)$ is concave. Also, $\mathcal{P}$ is convex and compact. Hence by Glicksberg's fixed point theorem~\cite{glicksberg-52} there exists a Nash equilibrium for the game $G$. The result now follows from Theorem~\ref{thm:allocation-reduction} and the interchangeability property of Nash equilibria in two-player zero-sum games (see~\cite[Theorem~3.1]{tijs-game-theory}).
\end{proof}

\begin{proof}[Proof of Theorem~\ref{thm:weighted}]
A local c-optimal design that minimises the variance of the estimator $\sum_{d \in [D]}n_d (u^T\hat{p}(d))$ is obtained by finding a solution $v_{t,d}^*, t \in T, d = 1, \ldots, D$, to the following optimisation problem:
\begin{align}
\text{minimise} & \qquad  \sum_{d =1}^D n_d^2 u^T  \left( \sum_{t \in T} v_{t,d} I_t(p(d))/c_t \right)^{-1} u  \label{eqn:objective-reduction-weighted}\\
\text{subject to} & \qquad \sum_{t \in T} \sum_{d=1}^D v_{t,d} \leq 1, \quad v_{t,d} \geq 0, \, t \in T, \, d = 1, \ldots , D. \nonumber
\end{align}
With additional variables $(m_d, d = 1,\ldots, D)$ which can be interpreted as the budget fractions given to each district, the value of the above optimisation problem is equal to the value of the following problem:
\begin{align}
\text{minimise} & \qquad  \sum_{d =1}^D n_d^2 \times u^T  \left( \sum\limits_{t \in T} v_{t,d} I_t(p(d))/c_t \right)^{-1} u  \label{eqn:objective-reduction-weighted-1}\\
\text{subject to} & \qquad \sum\limits_{t \in T}  v_{t,d} \leq m_d, d = 1, \ldots, D, \nonumber \\
& \qquad v_{t,d} \geq 0, \, t \in T, \, d = 1, \ldots , D,  \nonumber \\
& \qquad \sum_{d=1}^D m_d  \leq 1, \nonumber \\
& \qquad m_d \geq 0, d = 1, \ldots, D \nonumber.
\end{align}
For a given $(m_d, d = 1, \ldots D)$, the optimisation in \eqref{eqn:objective-reduction-weighted-1} over the variables $(v_{t,d}, t \in T, d = 1, \ldots, D)$ can be performed separately for each $d$. {As argued in the proof of Theorem~\ref{thm:allocation-reduction}, using assumption~\ref{assm:pd} for each $p(d), d=1,\ldots, D$, there exists a solution to the problem~(\ref{eqn:objective-reduction}) with $p(d)$ in place of $p$; we denote it by $v^*(p(d))$ and the corresponding value by $a(v^*(p(d)); p(d))$}. Hence the above problem reduces to
\begin{align}
\text{minimise} & \qquad  \sum_{d =1}^D \frac{n_d^2}{m_d} a(v^*(p(d));p(d)) \label{eqn:objective-reduction-weighted-2}\\
\text{subject to} & \qquad \sum_{d=1}^D m_d  \leq 1, \nonumber \\
& \qquad m_d \geq 0, d = 1, \ldots, D \nonumber.
\end{align}
Note that \eqref{eqn:objective-reduction-weighted-2} is a convex problem in the variables $(m_d, d=1,\ldots, D)$; indeed the mapping $(m_d, d=1,\ldots, D) \mapsto \sum_{d \in [D]} \frac{n_d^2}{m_d} a(v^*(p(d));p(d))$ is lower-semicontinuous, convex and the constraint set is a compact and convex subset of $\mathbb{R}_+^D$. Hence, there exists a solution $(m_d^*, d=1, \ldots, D)$ to the above problem. Consider the Lagrangian
\begin{align*}
L(m, \lambda, \mu) =  \sum_{d =1}^D \frac{n_d^2}{m_d} a(v^*(p(d));p(d)) - \sum_{d=1}^d \lambda_d m_d + \mu\left(\sum_{d=1}^D m_d-1\right).
\end{align*}
By the Karush-Kuhn-Tucker conditions \cite[Chapter~5]{bvop}, there exist non-negative numbers  $(\lambda_d^*, d=1,\ldots, D)$ and $\mu^*$ such that the following conditions hold:
\begin{align}
-\frac{n_d^2 a(v^*(p(d)); p(d))}{(m_d^*)^2} - \lambda_d^* + \mu^* & =0, d=1, \ldots, D, \label{eqn:kkt-objective-reduction-weighted-2}\\
\lambda_d^* m_d^* & = 0, d = 1, \ldots, D, \label{eqn:slackness-objective-reduction-weighted-2-lambda} \\
\mu^*\left(\sum_{d=1}^D m_d^* - 1 \right) &= 0 \label{eqn:slackness-objective-reduction-weighted-2}.
\end{align}
If $\mu^* = 0$, then~(\ref{eqn:kkt-objective-reduction-weighted-2}) is violated for all $d$; hence $\mu^* > 0$ and by~(\ref{eqn:slackness-objective-reduction-weighted-2}), $\sum_{d \in [D]} m^*_d = 1$. Whenever $m_d^* >0$, we must have $\lambda_d^* =0$, and by~(\ref{eqn:kkt-objective-reduction-weighted-2}),  $m_d^* = n_d\sqrt{\mu^* a(v^*(p(d)); p(d))}$. Using $\sum_{d=1}^D m_d^* = 1$, we can solve for $\mu^*$ to get
\begin{equation}
\label{eqn:md*}
m_d^* = \frac{n_d \sqrt{a(v^*(p(d)); p(d))}}{\sum_{d'=1}^D n_{d'}\sqrt{a(v^*(p(d')); p(d'))}}.
\end{equation}
This solves problem \eqref{eqn:objective-reduction-weighted-2}. By setting
$$
v^*_{t,d} = v^*_t(p(d)) m_d^*, \quad t \in T, d=1,\ldots, D
$$
where $m_d^*$ is as given in \eqref{eqn:md*}, we also solve problem \eqref{eqn:objective-reduction-weighted-1}. By the equivalence of the problems \eqref{eqn:objective-reduction-weighted} and \eqref{eqn:objective-reduction-weighted-1}, we have a solution to problem \eqref{eqn:objective-reduction-weighted}, which can be described as follows: Allocate $C_d = C m_d^*$ to district $d$, and by Theorem \ref{thm:allocation-reduction}, allocate $w^*_t(d) = v^*_t(p(d)) C_d/c_t$ for test pattern $t$ in district $d$. This completes the proof.
\end{proof}

\begin{proof}[Proof of Theorem~\ref{thm:allocation-symptomatic-model}]
The proof is straightforward and follows the same arguments used in the proof of Theorem~\ref{thm:weighted}; we only have to use $s$ in place of $d$ and recognise that the Fisher information matrices may depend on $s$.
\end{proof}

{\bf Acknowledgements:}
{SA's research was supported in part by MATRICS. GRB's research is funded by an Intermediate Fellowship by the Wellcome Trust DBT India Alliance (Clinical and Public Health Research Fellowship); grant number: IA/CPHI/14/1/501499. AI and NKV were supported by the IISc-Cisco Centre for Networked Intelligence (CNI), Indian Institute of Science, through a Google grant. MM and SS were supported by the CNI through a Hitachi grant. NR and SY works were supported by CNI through a Cisco-IISc PhD Research Fellowship. RS's work was supported by the Google Gift Grant. RS's work was done while on sabbatical leave at Strand Life Sciences, and RS gratefully acknowledges discussions with and feedback from Vamsi Veeramachaneni and Ramesh Hariharan.}

\bibliographystyle{abbrv}
\bibliography{OptimalTestingv3}

\end{document}